\documentclass[letterpaper, 10 pt, conference]{ieeeconf} 
 
\IEEEoverridecommandlockouts                              
 
\usepackage[utf8]{inputenc}  
\usepackage{mathtools}
\mathtoolsset{showonlyrefs,showmanualtags}
\usepackage{amssymb,amsthm}
\usepackage{dsfont}
\usepackage{booktabs}
\usepackage{color}
\usepackage{bbold} 
\usepackage[final]{graphicx}
\usepackage[yyyymmdd,hhmmss]{datetime}
\usepackage{bm,upgreek}
\usepackage{tablefootnote}
\usepackage{nth}
\usepackage{cite}

\newtheorem{Assumption}{Assumption}
\newtheorem{Theorem}{Theorem}

\newtheorem{Remark}{Remark}

\newtheorem{Lemma}{Lemma}

\newtheorem{Proposition}{Proposition}
\newtheorem{Problem}{Problem}

\DeclareMathOperator{\VEC}{vec}

\DeclareMathOperator{\COV}{Cov}
\newcommand{\Exp}[1]{\mathbb{E}[ #1]}
\newcommand{\Prob}[1]{\mathbb{P}(#1)} 
\newcommand{\ID}[1]{ \mathbb{1} (#1 ) }

\newcommand{\Compress}{\medmuskip=0mu
\thinmuskip=0mu
\thickmuskip=0mu}

\usepackage{tikz}
 \usetikzlibrary{hobby}
 \usetikzlibrary{patterns}
\usetikzlibrary{graphs,graphs.standard}
\usetikzlibrary{shapes,snakes}

\tikzset{%
  every neuron/.style={
    circle,
    draw,
    minimum size=.7cm
  },
  neuron missing/.style={
    draw=none, 
    scale=3,
    text height=0.333cm,
    execute at begin node=\color{black}$\vdots$
  },
}

\title{Deep Structured Teams in Arbitrary-Size Linear
Networks: Decentralized  Estimation, Optimal Control and Separation Principle}

\author{Jalal~Arabneydi and Amir G. Aghdam
 \thanks{
This work was supported in part by funding from the Innovation for  
Defence Excellence and Security (IDEaS) program from the Department of  
National Defence (DND). Any opinions and conclusions in this work are  those of the authors and do not reflect the views,  
positions, or policies of - and are not endorsed by - IDEaS, DND, or  
the Government of Canada.
}  
\thanks{Jalal Arabneydi, and Amir G. Aghdam are with the  Department of Electrical and Computer Engineering, 
        Concordia University, 1455 de Maisonneuve Blvd. West, Montreal, QC, Canada, Postal Code: H3G 1M8.  Email: {\tt\small jalal.arabneydi@mail.mcgill.ca},        
        {\tt\small aghdam@ece.concordia.ca}}%
}

\begin{document}

\maketitle
\thispagestyle{empty}
\pagestyle{empty}

\vspace*{-5cm}{\footnotesize{Proceedings of IEEE  Conference on Decision and Control, 2021.}}
\vspace*{4.2cm}

\begin{abstract}
In this article, we introduce decentralized Kalman filters for linear quadratic deep structured teams. The agents in deep structured teams are coupled in dynamics, costs and measurements through a set of linear regressions of the states and actions (also called deep states and deep actions). The information structure is decentralized, where every agent observes a noisy measurement of its local state and the global deep state. Since the number of agents is often very large in deep structured teams, any naive approach to finding an optimal Kalman filter suffers from the curse of dimensionality. Moreover, due to the decentralized nature of information structure, the resultant optimization problem is non-convex, in general, where non-linear strategies can outperform linear ones. However,   we prove that the optimal strategy is linear  in the local state estimate as well as the  deep state  estimate and can be efficiently  computed by two scale-free Riccati equations  and Kalman filters.    We propose a bi-level orthogonal approach across both space and time levels  based on a gauge transformation technique to achieve the above result.
 We  also establish a separation  principle between optimal control and optimal estimation.  Furthermore, we show that as the number of agents goes to infinity,  the Kalman gain associated with  the deep state estimate converges to zero  at a rate inversely proportional to the number of agents. This leads to a fully decentralized approximate strategy where every agent  predicts the deep state by  its conditional and unconditional expected value, also known as the certainty equivalence approximation and (weighted) mean-field approximation, respectively. 
\end{abstract}
\section{Introduction}

In recent years, there has been a growing interest in large-scale systems such as  social networks, epidemics,  smart grids, economics and  robotics, to name only a few. 
In such systems, a large number of decision-makers interact with each other in order to minimize a common cost function, where  the global decision process (i.e., system-level decision) is a manifestation of many different local decision processes (i.e., agent-level decisions). To be able to coordinate the global process based on the local processes, the standard  approach is to consider \emph{centralized information}.  On the other hand, centralized information is not desirable in large-scale systems due to physical and economic limitations. Hence in practice,  the information set of each local decision-maker is  often  different,  leading to a discrepancy in perspective.  Since  it is  conceptually  challenging to establish  coordination among  agents  with different viewpoints,    \emph{decentralized systems} do not typically  admit   a globally optimal  solution.

In general, information structures can be categorized into three classes: \textit{classical} (centralized),~\textit{partially nested} (semi-centralized) and \textit{non-classical} (decentralized). When every agent knows the history of the actions and observations of all agents, the information structure is called \emph{classical}. On the other hand,  when every agent knows the history of the actions and observations of all agents whose actions affect its observations, the information structure is called partially nested; 
any other information structure is called non-classical~\cite{Witsenhausen1971separation,ho1972team}. It is well-known that the optimal strategy is an affine function of the observations in linear quadratic Gaussian (LQG) systems under classical and partially-nested information structures. However, this is not  the case for non-classical information structures, in general. For example,  it is shown in~\cite{Witsenhausen1968Counterexample} that solving a simple two-agent  LQG model with decentralized information   ends up with   a non-convex optimization problem where non-linear strategies outperform linear ones.  Even if the attention is  restricted to linear strategies, the best linear solution is not necessarily a convex optimization problem solution, except for a few special cases such as quadratic invariance~\cite{rotkowitz2006characterization}. In addition,    the best linear strategy might not even have a finite-dimensional representation~\cite{whittle1974optimal}. For more counterexamples on the complexity of  decentralized linear problems, see~\cite{Lipsa2011counterexample,Yuksel2009counterexample}.

In this paper, we study a newly emergent class of large-scale decentralized control systems called \emph{deep structured teams}~\cite{Jalal2020CCTA,Jalal2021CDC_MPC, Jalal2019MFT,Jalal2019risk,Jalal2020Automatica_On,
Vida2020CDC,Masoud2020CDC,Jalal2020Nash}, which may be viewed as a generalization of the notion of weighted mean-field teams  introduced in~\cite{arabneydi2016new}.  Since the closest
model to such systems  is feed-forward deep neural networks, we refer to them as \emph{deep structured teams/games}. In particular, agents in deep structured models interact with each other through a set of linear regressions of the states and actions of all agents,  which is similar in spirit to  the interaction of  neurons in feed-forward deep neural networks.   In addition,  we show  in~\cite{Jalal2019risk} that the secret ingredient of finding a low-dimensional solution in such
models is related to \emph{invariance/equivariance symmetry}, which is the backbone of deep learning. Furthermore, we  demonstrate in~\cite{Jalal2021CDC_MPC} that today's  most common
 feed-forward deep neural networks (i.e., those with  rectified linear unit activation function) may be viewed as a special case
of deep structured teams,  where
layers are time steps and neurons are simple integrator agents whose goal is to collaborate in order to minimize a common loss (cost) function.  For more applications of deep structured models,  the  reader is referred to  reinforcement learning~\cite{Jalal2020CCTA,
    Vida2020CDC,Masoud2020CDC}, nonzero-sum game~\cite{Jalal2020Automatica_On,Jalal2020Nash},  minmax optimization~\cite{Jalal2019LCSS},  leader-followers~\cite{JalalCCECE2018,JalalCDC2018}, epidemics~\cite{JalalACC2019}, smart grids~\cite{JalalACC2018}, mean-field teams~\cite{JalalCDC2015, JalalCDC2017,JalalCDC2014,Jalal2017linear},   and networked  estimation~\cite{Jalal2019TSNE,Jalal2018fault}.

Herein, we focus on  LQG deep-structured systems  under a non-classical information structure that is neither partially nested nor quadratically invariant. In particular, we  introduce a \emph{bi-level orthogonal approach} to finding a low-dimensional solution using the gauge transformation technique, initially proposed in~\cite{arabneydi2016new}. More precisely, we prove that the optimal strategy is an affine function of the observations, where controllers' and observers' gains are computed  by two scale-free  local and  global Riccati equations  and Kalman filters, respectively. The derivation of  the proposed decentralized  Kalman filters is  different from the standard  one in~\cite{kalman1960new,simon2006optimal}.  For example, our proof technique holds only for i.i.d. random variables  and  naturally for those non-i.i.d.  variables that are \emph{conditionally} i.i.d. relative to some  latent variables. For instance, exchangeable random variables (that are not necessarily i.i.d.) behave as conditionally i.i.d. variables in  the infinite-population model according to De Finetti's theorem~\cite{diaconis1980finetti}. To the best of our knowledge, this is the first result establishing a tractable optimal strategy under noisy measurements for a class of large-scale control systems  with non-classical  information structure.  However, at this stage, it is unclear  whether or not such a low-dimensional representation exists for general non-i.i.d. random variables. 

The rest of the paper is organized as follows. In Section~\ref{sec:problem}, we formulate an LQG deep structured team problem with noisy observations. In Section~\ref{sec:estimation},  we  derive our Kalman filters  in five steps. In Section~\ref{sec:main}, we present the main result of the paper followed by  two extensions to infinite-population approximation and least square estimation. In Section~\ref{sec:conclusions},  we summarize and conclude the paper.

\section{Problem formulation}\label{sec:problem}
Throughout the paper, $\mathbb{R}$ and $\mathbb{N}$ refer to  the sets of real and natural numbers, respectively. Given any $n \in \mathbb{N}$, $\mathbb{N}_n=\{1,2,\ldots,n\}$ and  $x_{1:n}=\{x_1,\ldots,x_n\}$. For vectors $x, y, z$, $\VEC(x,y,z)=[x^\intercal, y^\intercal, z^\intercal]^\intercal$.  
$\COV(x)=\Exp{ (x-\Exp{x})  (x-\Exp{x})^\intercal}$ is the covariance matrix. $\mathcal{N}(\mu, \Sigma)$ is a multi-dimensional Gaussian probability distribution with mean $\mu$ and covariance matrix $\Sigma$. In addition, $\mathbf 0$, $\mathbf 1$, $\mathbf I$ represent  matrix zero where all arrays are zero, matrix one where all arrays are one, and identity matrix, respectively.  $\otimes$ is Kronecker product. Given  vector $\mathbf x=\VEC(x_1,\ldots,x_n) \in \mathbb{R}^{n}$, $\mathbf x^{-i} \in \mathbb{R}^{n-1}$ denotes  $\mathbf x$ without the $i$-th component.  

Consider a system consisting of $n \in \mathbb{N}_n$ agents. Let  $x^i_t \in \mathbb{R}^{d_x}$, $u^i_t \in \mathbb{R}^{d_u}$ and $w^i_t \in \mathbb{R}^{d_w}$, $d_x, d_u, d_w \in \mathbb{N}$, denote the state, action and noise of agent $i \in \mathbb{N}_n$  at time $t \in \mathbb{N}$, respectively.  For each agent $i \in \mathbb{N}_n$,  define
\begin{equation}\label{eq:deep_state}
\bar x_t:=\frac{1}{n}\sum_{i=1}^n \alpha_i x^i_t, \quad \bar u_t:=\frac{1}{n}\sum_{i=1}^n \alpha_i u^i_t, \quad \bar w_t:=\frac{1}{n}\sum_{i=1}^n \alpha_i w^i_t.
\end{equation}

Without loss of generality, we assume that  the influence factors are normalized as follows: $\frac{1}{n} \sum_{i=1}^n \alpha_i^2=1$.
 Let the dynamics of agent $i \in \mathbb{N}$ at time $t \in \mathbb{N}$ be described by
\begin{equation}\label{eq:dynamics}
x^i_{t+1}=A_t x^i_t + B_t u^i_t + E_t w^i_t+  \alpha_i (\bar A_t \bar x_t + \bar B_t \bar u_t + \bar E_t \bar w_t),
\end{equation}
where matrices $A_t$, $B_t$, $E_t$, $\bar A_t$, $\bar B_t$ and $\bar E_t$ have appropriate dimensions. Denote by $y^i_t \in \mathbb{R}^{d_y}$ and $v^i_t \in \mathbb{R}^{d_v}$ the state observation and measurement noise of agent $i$  at time $t$,  respectively, such that
\begin{equation}
y^i_t=C_t x^i_t +S_t v^i_t + \alpha_i(\bar C_t \bar x_t + \bar S_t \bar v_t),
\end{equation}
where $C_t$, $S_t$, $\bar C_t$ and $\bar S_t$ have appropriate dimensions and 
\begin{equation}\label{eq:noisy_deep_state}
\bar y_t:=\frac{1}{n}\sum_{i=1}^n \alpha_i y^i_t, \quad \bar v_t:=\frac{1}{n}\sum_{i=1}^n \alpha_i v^i_t.
\end{equation} 
Following the terminology of deep structured models~\cite{Jalal2020CCTA,Jalal2019risk}, aggregate variables (linear regressions) defined in~\eqref{eq:deep_state} and~\eqref{eq:noisy_deep_state} are called \emph{deep} variables due to the fact that their  evolutions across time  horizon  are similar to   feed-forward deep neural networks. Hence, for ease of reference,  we refer to $\bar x_t$ as \emph{deep state}  at time~$t$ in the sequel.

We consider a non-classical (decentralized) information structure  called \emph{imperfect deep state sharing} (IDSS) such that the action of agent $i$ is given by
\begin{equation}
u^i_t=g^i_t(y^i_{1:t}, \bar y_{1:t}), \tag{IDSS}
\end{equation}
where  $g^i_t:\mathbb{R}^{2 t d_y} \rightarrow \mathbb{R}^{d_u}$. When $n\geq 3$, a salient property of IDSS structure  is that it provides \emph{natural encryption} of data in terms of noisy deep state (linear regression)  so that no agent knows  the local (private) state of other agents. Let $\mathbf x_t:=\VEC(x^1_t,\ldots,  x^n_t)$, $ \mathbf u_t:=\VEC(u^1_t,\ldots, u^n_t)$, $\mathbf w_t:=\VEC(w^1_t,\ldots,  w^n_t)$, $ \mathbf y_t:=\VEC(y^1_t,\ldots, y^n_t)$, and  $\mathbf v_t:=\VEC(v^1_t,\ldots,  v^n_t)$. The random variables $\{\mathbf x_1, \mathbf w_{1:T}, \mathbf v_{1:T}\}$ are  defined on a common probability space and are mutually independent across agents and control horizon. Furthermore, $x^i_1 \sim \mathcal{N}(\mu^x, \Sigma^x)$, $w^i_t \sim \mathcal{N}(\mathbf 0, \Sigma^w_t)$, $v^i_t \sim \mathcal{N}(\mathbf 0, \Sigma^v_t)$, $t \in \mathbb{N}_T$.
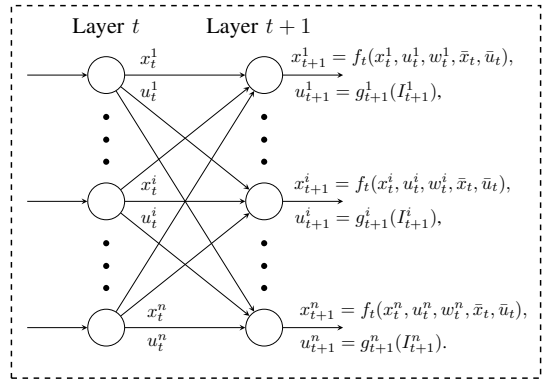
\begin{figure}[t!]
\hspace{0cm}
\scalebox{.7}{
\begin{tikzpicture}[x=1.5cm, y=1.2cm, >=stealth]
  \draw[font=\large](0,2.1cm) node {Layer $t$} ;
\draw[font=\large](2.9cm,2.1cm) node {Layer $t+1$} ;
\draw[font=\large](6.5cm,-5cm) node {} ;
\draw[black, thick, dashed] (-1.2,2.1) rectangle (5.5,-3.8);

\foreach \m/\l [count=\y] in {1,missing,2,missing,3}
  \node [every neuron/.try, neuron \m/.try] (input-\m) at (0,2-\y) {};

\foreach \m/\l [count=\y] in {1,missing,2,missing,3}
  \node [every neuron/.try, neuron \m/.try ] (output-\m) at (2,2-\y) {};

  \draw [<-] (input-1) -- ++(-1,0)
    node [above,midway] {\hspace{3.5cm}$x^1_t$}
        node [below,midway] {\hspace{3.5cm}$u^1_t$};
     \draw [<-] (input-2) -- ++(-1,0)
    node [above,midway] {\hspace{3.5cm}$x^i_t$}
        node [below,midway] {\hspace{3.5cm}$u^i_t$};
        
          \draw [<-] (input-3) -- ++(-1,0)
    node [above,midway] {\hspace{3.8cm}$x^{n}_t$}
        node [below,midway] {\hspace{3.8cm}$u^{n}_t$};

  \draw [->]  (output-1) -- ++(1,0)  node[above, midway]  
   {\hspace*{3.4cm}$x^1_{t+1}= f_t(x^1_t,u^1_t,w^1_t, \bar x_t, \bar u_t),$}
     node [below,midway] {\hspace{2.1cm}$u^1_{t+1} = g^1_{t+1}( I^1_{t+1})$,};   
      \draw [->]  (output-2) -- ++(1,0)  node[above, midway]     {\hspace*{3.4cm}$x^i_{t+1}= f_t(x^i_t,u^i_t, w^i_t, \bar x_t, \bar u_t),$}
     node [below,midway] {\hspace{2.1cm}$u^i_{t+1} = g^i_{t+1}( I^i_{t+1})$,};  
    
      \draw [->]  (output-3) -- ++(1,0)  node[above, midway]  { \hspace*{3.7cm} $x^{n}_{t+1}= f_t(x^{n}_t,u^{n}_t, w^{n}_t, \bar x_t, \bar u_t),$}
     node [below,midway] {\hspace{2.3cm}$u^{n}_{t+1} = g^{n}_{t+1}( I^{n}_{t+1})$.};     
     
    \foreach \i in {1,...,3}
  \foreach \j in {1,...,3}
    \draw [->] (input-\i) -- (output-\j);    
\end{tikzpicture}}
\vspace{-.4cm}
\caption{ The interaction (coupling)  between the agents in deep structured teams is similar in spirit  to that of neurons in a
 feed-forward deep neural network, where agents may be viewed as neurons and time steps as layers. In this paper, dynamics $f_t$ is an affine function and information set $I^i_t$ is imperfect deep state sharing.} 
\end{figure}
 Let $\mathbf g:=\{\{g^i_t\}_{i=1}^n\}_{t=1}^T$ denote the strategy of the system. We now define the team cost as follows.
\begin{equation}
J_n(\mathbf g):=\Exp{\frac{1}{n}\sum_{t=1}^T \sum_{i=1}^n c^i_t},
\end{equation}
where the per-step cost function of agent $i \in \mathbb{N}_n$ is
\begin{equation}\label{eq:per_step_cost}
c^i_t:=(x^i_t)^\intercal Q_t x^i_t+ (u^i_t)^\intercal R_t u^i_t+ \bar x_t^\intercal \bar Q_t \bar x_t + \bar u_t^\intercal \bar R_t \bar u_t.
\end{equation}
It is straightforward to consider cross terms of the form $\alpha_i x^i_t \tilde Q_t \bar x_t$ and  $\alpha_i u^i_t \tilde R_t \bar u_t$ in~\eqref{eq:per_step_cost}; see also~\cite[Remark 1]{Jalal2021CDC_MPC} for other straightforward extensions.

  We are interested in the following optimization problem.
\begin{Problem}\label{prob1}
Given any $n \in \mathbb{N}$, find the optimal strategy $\mathbf g^\ast$ such that for any strategy $\mathbf g$: $J_n(\mathbf g^\ast) \leq J_n(\mathbf g)$. 
\end{Problem}
\section{Derivation of Kalman filters}\label{sec:estimation}
Prior to delving into details, we present $3$ standard lemmas.

\begin{Lemma}\label{lemma:Gaussian_orth}
Let $x$ and $y$ be   Gaussian random variables with zero mean. The best non-linear strategy is  equal to the best linear one as far as  the minimum mean-square error  is concerned, i.e.,
$\Exp{x\mid y}=\Exp{x y^\intercal} \Exp{yy^\intercal}^{-1}y$.
In addition,
\begin{equation}
\Exp{ (x-\Exp{x\mid y})(x-\Exp{x\mid y})^\intercal }=\Exp{xx^\intercal}-\Exp{ x y^\intercal} \Exp{y y^\intercal} \Exp{ y x^\intercal}.
\end{equation}
Furthermore,  $x$ and $y$ are independent  if and only if $\Exp{x y^\intercal}=0$, i.e.  $x \perp y$.
\end{Lemma}

\begin{Lemma}\label{lemma:Summation}
Let $x$, $y_1$ and $y_2$ be Gaussian random variables. Let also $y_1$ and $y_2$ be independent and have  zero mean. Then,
\begin{equation}
\Exp{x \mid y_1,y_2}=\Exp{x \mid y_1} + \Exp{x \mid y_2}.
\end{equation}
\end{Lemma}
\begin{proof}
The proof follows from the fact that $\Exp{y_1 y_2^\intercal}=0$
\end{proof}
At any  time $t \in \mathbb{N}_T$, denote by $\mathcal{H}^g_t:=\{\mathbf y_{1:t}, \mathbf u_{1:t}\}$  and $\mathcal{H}_t:=\{\mathbf y_{1:t}, \mathbf u_{1:t-1}\}$  the  history sets with and without joint action~$\mathbf u_t$,  respectively.
\begin{Lemma}\label{lemma:independent_g}
  The conditional expectation of the joint state at time~$t \in \mathbb{N}_T$ given $\mathcal{H}^{\mathbf g}$ does not depend on the strategy $\mathbf g$, i.e.
$\Exp{\mathbf x_{t} \mid \mathcal{H}^g_t}=\Exp{\mathbf x_{t} \mid \mathcal{H}_t}$.
\end{Lemma}
\emph{Proof.}
The proof follows form the fact that the following equality holds irrespective  of strategy $\mathbf g$, i.e.
$
\Prob{\mathbf x_{t} \mid \mathcal{H}^g_t}=\Prob{\mathbf x_{t} \mid \mathcal{H}_t},
$
where from  Bayes' rule,
\begin{align}
&\Prob{\mathbf x_{t} \mid \mathcal{H}^g_t}=\frac{\Prob{\mathbf x_{t} , \mathbf u_t\mid \mathcal{H}_t}}{\int_{\tilde{\mathbf x}_t} \Prob{\mathbf{\tilde{x}}_{t}, \mathbf u_t\mid \mathcal{H}_t} d(\tilde{\mathbf x}_t)}\\
&=\frac{\Prob{\mathbf u_t \mid \mathcal{H}_t} \Prob{\mathbf x_{t} \mid \mathcal{H}_t}}{\int_{\tilde{\mathbf x}_t} \Prob{\mathbf u_t \mid \mathcal{H}_t} \Prob{\mathbf{\tilde{x}}_{t}\mid \mathcal{H}_t} d(\tilde{\mathbf x}_t)}\\
&=\frac{\ID{\mathbf u_t=\mathbf g_t(\mathcal{H}_t)} \Prob{\mathbf x_{t} \mid \mathcal{H}_t}}{\ID{\mathbf u_t=\mathbf g_t(\mathcal{H}_t)} \int_{\tilde{\mathbf x}_t} \Prob{\mathbf{\tilde{x}}_{t}, \mathbf u_t\mid \mathcal{H}_t} d(\tilde{\mathbf x}_t)}= \Prob{\mathbf x_{t} \mid \mathcal{H}_t}.
\end{align}
 We now take 5 steps to derive a low-dimensional solution.
\subsection{Step 1: Gauge transformation} 
Define the following auxiliary variables using the gauge transformation presented in~\cite{Jalal2019risk} such that for every $i \in \mathbb{N}_n$:
\begin{align}\label{eq:gauge}
\Delta x^i_t&:=x^i_t - \alpha_i \bar x_t, \hspace{.1cm} \Delta u^i_t:=u^i_t - \alpha_i \bar u_t, \hspace{.1cm} \Delta w^i_t:=w^i_t - \alpha_i \bar w_t, \nonumber \\
\Delta y^i_t&:=y^i_t - \alpha_i \bar y_t, \hspace{.1cm} \Delta v^i_t:=v^i_t - \alpha_i \bar v_t.
\end{align}

\begin{Lemma}[Linear dependence between auxiliary variables]\label{lemma:LD_aux}
The gauge transformation~\eqref{eq:gauge} introduces the following relations:  $\sum_{i=1}^n \alpha_i \Delta x^i_t=\mathbf{0}, \sum_{i=1}^n \alpha_i \Delta u^i_t=\mathbf{0}, \sum_{i=1}^n \alpha_i \Delta w^i_t=\mathbf{0}
$, $\sum_{i=1}^n \alpha_i \Delta y^i_t=\mathbf{0}$ and $\sum_{i=1}^n \alpha_i \Delta v^i_t=\mathbf{0}$.
\end{Lemma}
\begin{proof}
The proof directly follows from~\eqref{eq:deep_state},~\eqref{eq:noisy_deep_state} and~\eqref{eq:gauge}.
\end{proof}
From Lemma~\ref{lemma:LD_aux} and equations~\eqref{eq:deep_state} and~\eqref{eq:dynamics}, one arrives at:
\begin{align}\label{eq:dynamics_deep}
\bar x_{t+1}&=(A_t+\bar A_t)\bar x_t + (B_t+\bar B_t) \bar u_t + (E_t + \bar E_t) \bar w_t, \nonumber \\
\bar y_{t}&=(C_t+\bar C_t)\bar x_t + (S_t + \bar S_t) \bar v_t,
\end{align}
and for every $i \in \mathbb{N}_n$:
\begin{align}
\Delta x^i_{t+1}&=A_t \Delta x^i_t + B_t \Delta u^i_t + E_t \Delta w^i_t, \nonumber \\
\Delta y^i_{t}&=C_t \Delta x^i_t + S_t \Delta v^i_t.
\end{align}

\begin{Lemma}[Orthogonal relation in cost function]\label{lemma:cost}
Form the gauge transformation~\eqref{eq:gauge}, the per-step cost function~\eqref{eq:per_step_cost} can be expressed in terms of auxiliary variables and aggregate (deep) variables as follows:
\begin{multline}
\frac{1}{n} \sum_{i=1}^n \big( (x^i_t)^\intercal Q_t x^i_t +  (u^i_t)^\intercal R_t u^i_t \big)= \bar x_t^\intercal Q_t \bar x_t + \bar u_t^\intercal  R_t \bar u_t\\
+ \frac{1}{n} \sum_{i=1}^n \big((\Delta x^i_t)^\intercal Q_t \Delta x^i_t
 +  (\Delta u^i_t)^\intercal R_t \Delta u^i_t\big).
\end{multline}
\end{Lemma}
\begin{proof}
From~\eqref{eq:deep_state} and~\eqref{eq:gauge},  $x^i_t= \Delta x^i_t + \alpha_i \bar x_t$ and  $u^i_t= \Delta u^i_t + \alpha_i \bar u_t$, $i \in \mathbb{N}_n$, can be replaced by the auxiliary and deep variables in the cost function. The proof follows from the fact that $\sum_{i=1}^n \alpha_i \Delta x^i_t Q_t \bar x_t=0$ and $\sum_{i=1}^n \alpha_i \Delta u^i_t R_t \bar u_t=0$.
\end{proof}
\begin{Lemma}[Relations between primitive random  variables]\label{lemma:Relations_primitive_rv}
The following holds for every finite $n \in \mathbb{N}$, $t\in \mathbb{N}$, and $i \neq j \in \mathbb{N}_n$:
\begin{enumerate}
\item $\Delta x^i_1 \not\perp \Delta x^j_1$ and
$\Exp{\Delta x^i_1 (\Delta x^j_1)^\intercal}= -\frac{\alpha_i\alpha_j}{n} \Sigma^x \neq \mathbf 0$,

\item $\Delta w^i_t \not\perp \Delta w^j_t$ and
$\Exp{\Delta w^i_t (\Delta w^j_t)^\intercal}= -\frac{\alpha_i\alpha_j}{n} \Sigma_t^w \neq \mathbf 0$,
\item $\Delta v^i_t \not\perp \Delta v^j_t$ and
$\Exp{\Delta v^i_t (\Delta v^j_t)^\intercal}= -\frac{\alpha_i\alpha_j}{n} \Sigma_t^v \neq \mathbf 0$,

\item $x^i_1 \not\perp \bar x_1$ and
$\Exp{x^i_1 (\bar x_1)^\intercal}= \frac{\alpha_i}{n} \Sigma^x \neq \mathbf 0$,

\item $w^i_t \not\perp \bar w_t$ and
$\Exp{w^i_t (\bar w_t)^\intercal}= \frac{\alpha_i}{n} \Sigma_t^w \neq \mathbf 0$,

\item $v^i_t \not\perp \bar v_t$ and
$\Exp{v^i_t (\bar v_t)^\intercal}= \frac{\alpha_i}{n} \Sigma_t^v \neq \mathbf 0$,

\item $\Exp{\Delta x^i_1 (\Delta x^i_1)^\intercal}=(1 -\frac{\alpha_i^2}{n}) \Sigma^x  \neq \mathbf 0$,

\item $\Exp{\Delta w^i_t (\Delta w^i_t)^\intercal}=(1 -\frac{\alpha_i^2}{n}) \Sigma_t^w  \neq \mathbf 0$,

\item  $\Exp{\Delta v^i_t (\Delta v^i_t)^\intercal}=(1 -\frac{\alpha_i^2}{n}) \Sigma_t^v  \neq \mathbf 0 $,

\item $\Delta x^i_1 \perp \bar x_1$ and $\Exp{\Delta x^i_1 (\bar x_1)^\intercal}= \mathbf 0$,

\item $\Delta w^i_t \perp \bar w_t$ and $\Exp{\Delta w^i_t (\bar w_t)^\intercal}= \mathbf 0$,

\item $\Delta v^i_t \perp \bar v_t$ and $\Exp{\Delta v^i_t (\bar v_t)^\intercal}= \mathbf 0$.
\end{enumerate}
The non-orthogonal relations (1)--(6) simplify to the orthogonal ones, as $n \rightarrow \infty$.
\end{Lemma}
\begin{proof}
The proof follows from~\eqref{eq:deep_state},~\eqref{eq:noisy_deep_state},~\eqref{eq:gauge} and the fact that driving noises and measurement noises are i.i.d. random vectors with zero mean.
\end{proof}

In the perfect sharing and deep state sharing, the certainty equivalence principle simplifies the analysis and results in  two standard decoupled Riccati equations~\cite{Jalal2020CCTA,Jalal2019risk}. In the imperfect observation case, however, the certainty equivalence principle does not hold. To see this, notice that although the dynamics and cost of the auxiliary subsystems are decoupled, their uncertainties are coupled (correlated) as shown in Lemma~\ref{lemma:Relations_primitive_rv}. This makes the analysis  more difficult.


\subsection{Step 2: Innovation processes}
Define the following global variables:
\begin{align}\label{eq:global_etimates}
z_{t|t}&:=\Exp{\bar x_t \mid \mathcal{H}_t},\quad &z_{t+1|t}&:=\Exp{\bar x_{t+1} \mid \mathcal{H}^g_t}, \nonumber \\
\xi_{t|t}&:=\bar x_t - z_{t|t},\quad  &\xi_{t+1|t}&:=\bar x_{t+1}- z_{t+1|t}, \nonumber \\
\bar \Sigma_{t|t}&:=\Exp{ \xi_{t|t} \xi_{t|t}^\intercal |\mathcal{H}_t}, \quad &\bar \Sigma_{t+1|t}&:=\Exp{ \xi_{t+1|t} \xi_{t+1|t}^\intercal |\mathcal{H}^g_t}.
\end{align}
From Lemma~\ref{lemma:independent_g} and equations~\eqref{eq:deep_state},~\eqref{eq:dynamics} and~\eqref{eq:global_etimates},  the dynamics of the deep-state estimate $z_{t+1|t}$ is given by
\begin{multline}{\label{eq:deep_state_est_dynamics}}
z_{t+1|t}=\Exp{\bar x_{t+1}\mid \mathcal{H}_t^g}=\Exp{(A_t+ \bar A_t) \bar x_t +(B_t +\bar B_t) \bar u_t \\+(E_t +\bar E_t)\bar w_t \mid \mathcal{H}^g_t}
=(A_t+ \bar A_t) z_{t|t} +(B_t +\bar B_t) \bar u_t.
\end{multline}
We now define local  variables for every $i$:
\begin{align}\label{eq:local_etimates}
&\Delta \hat x^i_{t|t}:=\Exp{\Delta x^i_t \mid \mathcal{H}_t}, \qquad   \Delta \hat x^i_{t+1|t}:=\Exp{\Delta x^i_{t+1} \mid \mathcal{H}^g_t}, \nonumber \\
&e^i_{t|t}:=\Delta  x^i_t - \Delta \hat x^i_{t|t},\qquad \hspace{.2cm} e^i_{t+1|t}:=\Delta  x^i_{t+1}- \Delta \hat x^i_{t+1|t},\nonumber \\
& \Sigma^{i}_{t|t}:=\Exp{e^i_{t|t} (e^i_{t|t})^\intercal |\mathcal{H}_t}, \quad \Sigma^i_{t+1|t}:=\Exp{ e^i_{t+1|t} (e^i_{t+1|t})^\intercal |\mathcal{H}^g_t}.
\end{align}

From Lemma~\ref{lemma:independent_g} and~\eqref{eq:deep_state},~\eqref{eq:dynamics},~\eqref{eq:gauge} and~\eqref{eq:local_etimates}, it results that 
\begin{align}
\Delta \hat x^i_{t+1|t}&=\Exp{A_t \Delta x^i_t+ B_t \Delta u^i_t+ E_t\Delta w^i_t\mid \mathcal{H}_t^g}\\
&=A_t \Delta \hat x^i_t+B_t \Delta u^i_t.
\end{align}

We define the following \emph{local innovation processes} for any $i \in \mathbb{N}_n$ and time $t \in \mathbb{N}_{T-1}$:
\begin{equation}
p^i_{t+1}:= y^i_{t+1}-\Exp{y^i_{t+1} \mid \mathcal{H}_t^g}.
\end{equation}
We decompose the above innovation processes according to the gauge transformation such that 
\begin{equation}\label{eq:innovation_pro_def}
\begin{cases}
\Delta p_{t+1}^i:= p^i_{t+1} -\alpha_i \bar p_{t+1}= \Delta y^i_{t+1}-\Exp{\Delta y^i_{t+1} \mid \mathcal{H}_t^g}, \\
\bar p_{t+1}:=\frac{1}{n} \sum_{i=1}^n \alpha_i p_{t+1}^i= \bar y_{t+1}-\Exp{\bar y_{t+1} \mid \mathcal{H}_t^g},\\
p^i_{t+1}=\Delta p^i_{t+1} + \alpha_i \bar p^i_{t+1}.
\end{cases}
\end{equation}
In the sequel, we refer to $\Delta p^i_t$ as \emph{the auxiliary innovation process} of agent $i$ and to $\bar p_t$ as  \emph{the  global innovation process}.

\begin{Remark}
\emph{Note that the following non-orthogonal relations hold for every finite $n \in \mathbb{N}$ and $i \neq j \in \mathbb{N}_n$:
\begin{enumerate}
\item $\Delta p^i_{t+1} \not\perp \Delta p^j_{t+1}$,
\item $p^i_{t+1} \not\perp \bar p_{t+1}$,
\item $p^i_{t+1} \not\perp p^j_{t+1}$.
\end{enumerate}
The above  non-orthogonal relations in (1)--(3) simplify to the orthogonal ones, as $n \rightarrow \infty$.}
\end{Remark}

\begin{Lemma}[Orthogonality in the transformed space]\label{lemma:Orthogonality_ innovation_processes}
Given $\mathcal{H}^g_t$, the following holds for every  $n \in \mathbb{N}$ and $i \in \mathbb{N}_n$:
\begin{enumerate}

\item $\Delta p^i_{t+1}  \perp \bar p_{t+1}$,

\item $\Delta p^i_{t+1}  \perp \bar x_{t+1}$,

\item $\Delta x^i_{t+1}  \perp \bar p_{t+1}$.

\end{enumerate}
\end{Lemma} 
\begin{proof}
We first prove that 
$\{\Delta x^i_{t+1},\Delta y^i_{t+1} \mid \forall i \in \mathbb{N}_n\}$ and  $\{\bar x_{t+1},\bar y_{t+1}\}$ are two independent sets. In particular,  given the history set $\mathcal{H}_t$,  we  show that the randomness of $\Delta x^i_{t+1}$ and $\Delta y^i_{t+1}$ is completely characterized by  the set $\{\Delta x^i_1,\Delta w^i_{1:t}, \Delta v^i_{t+1}\}$ as follows:
\begin{equation}
\begin{cases}
\Delta x^i_{t+1}= (\prod_{k=1}^t A_k) \Delta x^i_1+ \sum_{k=1}^t \mathcal{A}(t-k) B_k \Delta u^i_k \\
\qquad \qquad +  \sum_{k=1}^t \mathcal{A}(t-k) E_k \Delta w^i_k,\\
\Delta y^i_{t+1}=C_{t+1}\Delta x^i_{t+1}+ S_{t+1} \Delta v^i_{t+1},
\end{cases}
\end{equation}
where $\mathcal{A}(0):=\mathbf{I}$ and $\mathcal{A}(t-k):=\prod_{k'=t-k}^{t-1} A_{k'}$. Similarly,   the randomness of $\bar x_{t+1}$ and $\bar y_{t+1}$ is completely characterized by  the set $\{\bar x_1,\bar w_{1:t}, \bar v_{t+1}\}$ such that
\begin{equation}
\begin{cases}
\bar x_{t+1}\hspace{-.1cm}=\hspace{-.1cm} (\prod_{k=1}^t (A_k+ \bar A_k)) \bar x_1\hspace{-.1cm}+\hspace{-.1cm} \sum_{k=1}^t \bar{\mathcal{A}}(t-k) (B_k+\bar B_k) \bar u_k \\
\qquad \qquad +  \sum_{k=1}^t \bar{\mathcal{A}}(t-k) (E_k+\bar E_k) \bar w_k,\\
\bar y_{t+1}=(C_{t+1}+\bar C_{t+1})\bar x_{t+1}+ (S_{t+1}+\bar S_{t+1}) \bar v_{t+1},
\end{cases}
\end{equation}
where $\bar{\mathcal{A}}(0):=\mathbf{I}$ and $\bar{\mathcal{A}}(t-k):=\prod_{k'=t-k}^{t-1} (A_{k'}+\bar A_{k'})$. Therefore,  $\{\Delta x^i_{t+1},\Delta y^i_{t+1} \mid \forall i \in \mathbb{N}_n\}$ and  $\{\bar x_{t+1},\bar y_{t+1}\}$ are mutually independent  because their corresponding random sets are independent according to  properties 5-7 in Lemma~\ref{lemma:Relations_primitive_rv}.
 
We now prove the first property of the lemma.  From~\eqref{eq:innovation_pro_def}, it follows that
\begin{align}
&\Exp{\Delta p^i_{t+1} \bar p_{t+1}^\intercal \mid \mathcal{H}_t}= \mathbb{E}[(\Delta y^i_{t+1} -\Exp{\Delta y^i_{t+1}\mid \mathcal{H}_t})
 (\bar y_{t+1} \\
& -\Exp{\bar y_{t+1}\hspace{-.1cm} \mid \hspace{-.1cm} \mathcal{H}_t})^\intercal \mid \mathcal{H}_t]
\hspace{-.1cm}=\hspace{-.1cm} \Exp{(\Delta y^i_{t+1} \hspace{-.1cm}-\hspace{-.1cm} \Exp{\Delta y^i_{t+1}\mid \mathcal{H}_t}) \bar y_{t+1}^\intercal \mid \mathcal{H}_t } \\
&-  \Exp{(\Delta y^i_{t+1} -\Exp{\Delta y^i_{t+1}\mid \mathcal{H}_t}) }\Exp{\bar y_{t+1}\mid \mathcal{H}_t})^\intercal
\substack{(a)\\=} \mathbf 0,
\end{align} \hspace{-.1cm}
where $(a)$ follows from the fact that $\Exp{\Delta y^i_{t+1}\bar y_{t+1}^\intercal\mid \mathcal{H}_t}=  \Exp{\Delta y^i_{t+1}\mid \mathcal{H}_t} \Exp{\bar y_{t+1}\mid \mathcal{H}_t}^\intercal$ due their mutual independence established above. The proof is  completed from Lemma~\ref{lemma:independent_g}. A Similar argument holds for properties 2 and~3. 
\end{proof}

\subsection{Step 3: Covariance matrices}

From~\eqref{eq:deep_state},~\eqref{eq:dynamics} and~\eqref{eq:deep_state_est_dynamics},  the estimation error $\xi_{t+1\mid t}$ and covariance matrix $\bar \Sigma_{t+1|t}$  evolve at time $t \in \mathbb{N}_T$ as follows:
\begin{equation}
\begin{cases}
\xi_{t+1\mid t}=(A_t+\bar A_t) \xi_{t\mid t}+ (E_t+\bar E_t) \bar w_t,\\
\bar \Sigma_{t+1|t}=(A_t+\bar A_t) \bar \Sigma_{t|t} (A_t+\bar A_t)^\intercal+ \COV((E_t+\bar E_t) \bar w_t).
\end{cases}
\end{equation}
Similarly, the dynamics of  local  estimation error $e^i_{t|t}$ and covariance matrix $\Sigma^i_{t|t}$ for every $i \in \mathbb{N}_n$ at time $t \in \mathbb{N}_T$  are
\begin{equation}\label{eq:sigma_yes_i}
\begin{cases}
e^i_{t+1\mid t}=A_t e^i_{t\mid t}+ E_t \Delta w^i_t,\\
 \Sigma^i_{t+1|t}=A_t \Sigma^i_{t|t} A_t^\intercal+ \COV(E_t \Delta w^i_t).
 \end{cases}
\end{equation}

Given any $i \in \mathbb{N}_n$, one can define \textit{index-invariant} covariance matrices  as follows:
\begin{equation}\label{eq:sigma_not_i}
\Sigma_{t+1 \mid t}:=(1 - \frac{\alpha_i^2}{n})^{-1} \Sigma^i_{t+1 \mid t},\quad  \Sigma_{t \mid t}:= (1 - \frac{\alpha_i^2}{n})^{-1} \Sigma^i_{t \mid t}
\end{equation}
where
\begin{equation}
 \Sigma_{t+1|t}=A_t \Sigma_{t|t} A_t^\intercal+ E_t \Sigma^w_t E_t^\intercal.
\end{equation}
Denote by $\Delta \mathbf P^{-i}_{t+1} $ and $\boldsymbol \alpha^{-i}$ the joint vector consisting of  $\Delta p^j_{t+1}$ and $\alpha_j$, $\forall j \neq i \in \mathbb{N}_n$, respectively.
\begin{Lemma}\label{lemma:cross_terms}
For any $ i \neq j  \in \mathbb{N}_n$, the following equalities hold at any time $t \in \mathbb{N}_T$:
\begin{enumerate}
\item $\Exp{e^i_{t+1|t} (e^i_{t+1|t})^\intercal}= (1 - \frac{\alpha_i^2}{n}) \Sigma_{t+1|t}$,

\item $\Exp{e^i_{t+1 \mid t} (e^j_{t+1 \mid t})^\intercal}=\frac{-\alpha_i \alpha_j}{n - \alpha_i^2} \Sigma^{i}_{t+1 \mid t}=\frac{-\alpha_i \alpha_j}{n} \Sigma_{t+1|t}$,

\item $\Exp{\Delta x^i_{t+1} (\Delta \mathbf{p}^{-i}_{t+1})^\intercal}=\frac{-\alpha_i}{n}(\Sigma_{t+1 |t} C_{t+1}^\intercal) \otimes (\boldsymbol \alpha^{-i})^\intercal$,

\item $\Exp{\Delta \mathbf{p}^{-i}_{t+1} (\Delta \mathbf{p}^{-i}_{t+1})^\intercal}=(\mathbf I_{(n-1) \times (n-1)}  -\frac{1}{n}  \boldsymbol \alpha^{-i} (\boldsymbol \alpha^{-i})^\intercal ) \\
 \otimes (C_{t+1} \Sigma_{t+1|t}C_{t+1}^\intercal
+ S_{t+1} \Sigma^v_{t+1}S_{t+1}^\intercal)$.
\end{enumerate}

\end{Lemma}
\begin{proof}
The proof follows from~\eqref{eq:innovation_pro_def}--\eqref{eq:sigma_not_i} and Lemma~\ref{lemma:Relations_primitive_rv}.
\end{proof}
\begin{Lemma}
For each agent $i \in \mathbb{N}_n$, the following relation holds between global and local innovation processes
\begin{align}
\COV(p^i_{t+1})= \COV(\Delta p^i_{t+1}) + \alpha_i^2 \COV(\bar p_{t+1}).
\end{align}
\end{Lemma}
\begin{proof}
The proof follows from orthogonality in Lemma~\ref{lemma:Orthogonality_ innovation_processes} where
$\COV(y^i_{t+1} - \Exp{y^i_{t+1} \mid \mathcal{H}_t^g})=\COV(\Delta y^i_{t+1} -\Exp{\Delta y^i_{t+1}\mid \mathcal{H}^g_t}) + \alpha_i^2 \COV(\bar y_{t+1} -\Exp{\bar y_{t+1}\mid \mathcal{H}^g_t})$.
\end{proof}

\begin{Proposition}(A matrix-inversion for  other agents' auxiliary variables)\label{lemma:inverse}
Let $\boldsymbol \alpha=\VEC(\alpha_1,\ldots,\alpha_n) \in \mathbb{R}^{n}$ be any $n$-dimensional vector where   $\alpha_i \neq 0$, $\forall i \in \mathbb{N}_n$, and  $\frac{1}{n} \boldsymbol \alpha^\intercal \boldsymbol \alpha=\frac{1}{n} \sum_{i=1}^n (\alpha_i)^2=1$. Then, for every $i \in \mathbb{N}_n$,  it follows that
\begin{multline}
\big(\mathbf I_{(n-1) \times (n-1)} -\frac{1}{n} (\boldsymbol \alpha^{-i})(\boldsymbol \alpha^{-i})^\intercal\big)^{-1}\\
=\big(\mathbf I_{(n-1) \times (n-1)} + (\alpha_i)^{-2} (\boldsymbol \alpha^{-i})(\boldsymbol \alpha^{-i})^\intercal\big).
\end{multline}
\end{Proposition}
\begin{proof}
After some algebraic manipulations, one can show that the  multiplication  of the right- and left-hand sides result in an identity matrix.
\end{proof}

\subsection{Step 4: Updates}

Define the information set of the innovation processes at time $t+1$ as follows:
\begin{align}
{\mathcal H}^{\perp}_{t+1}:&=\{y^i_{t+1} -\Exp{y^i_{t+1}\mid \mathcal{H}_t^g}, \forall i \in \mathbb{N}_n \}\\
&=\{\Delta p^i_{t+1}, \forall i \in \mathbb{N}_n \} \cup  \{\bar p_{t+1}\}.
\end{align}
It is important to mention that 
$x \perp y$ (i.e. $\Exp{x y^\intercal}=\mathbf 0$), $\forall x \in \mathcal H^{\perp}_{t+1}, \forall y \in \mathcal{H}_t^g$,
where  the innovation processes at time $t+1$ provide an orthogonal decomposition from $\mathcal{H}^{g}_t$ by construction, on  noting that measurement and driving noises  at  $t+1$ have zero mean  and are  independent  of the history set $\mathcal{H}^g_t$. Thus,  history set $\mathcal{H}_{t+1}$ can be decomposed across time horizon into two orthogonal sets such that
\begin{equation}
\mathcal{H}_{t+1}=\mathcal{H}_t^g \cup \mathcal{H}^{\perp}_{t+1}.
\end{equation}

From Lemmas~\ref{lemma:Summation} and~\ref{lemma:independent_g}, we obtain 
\begin{align}\label{eq:update_bar_z}
z_{t+1\mid t+1}&=\Exp{\bar x_{t+1}\mid \mathcal{H}_{t+1}}
=\Exp{\bar x_{t+1}\mid \mathcal{H}_{t}}+ \Exp{\bar x_{t+1} \mid  \mathcal{H}^{\perp}_{t+1}} \nonumber \\
&=z_{t+1\mid t}+ \Exp{\bar x_{t+1} \mid  \mathcal{H}^{\perp}_{t+1}}.
\end{align}

\begin{Lemma}[Global update]\label{lemma:global} The update of the estimate of the deep state given $\mathcal{H}^{\perp}_{t+1}$ can be computed as follows:
\begin{align}
& \Exp{\bar x_{t+1} \mid  \mathcal{H}^{\perp}_{t+1}}
=\Exp{\bar x_{t+1} \mid  \bar p_{t+1}}=\bar \Sigma_{t+1 \mid t} (C_{t+1}+\bar C_{t+1})^\intercal\\
&\quad \times \Big( (C_{t+1}+\bar C_{t+1}) \bar \Sigma_{t+1 \mid t} (C_{t+1}+\bar C_{t+1})^\intercal \\
&\quad + (S_{t+1}+\bar S_{t+1}) \COV(\bar v_{t+1})(S_{t+1}+\bar S_{t+1})^\intercal \Big)^{-1}\\
&\quad \times (\bar y_{t+1}- (C_{t+1}+\bar C_{t+1}) z_{t+1 \mid t}).
\end{align}
\end{Lemma}
\begin{proof}
From Lemmas~\ref{lemma:Gaussian_orth} and~\ref{lemma:Orthogonality_ innovation_processes}, and equation~\eqref{eq:update_bar_z}, one has 
\begin{equation}
\Compress
\Exp{\bar x_{t+1} \mid  \mathcal{H}^{\perp}_{t+1}}=\Exp{\bar x_{t+1} \mid  \bar p_{t+1}}=\Exp{\bar x_{t+1} \bar p_{t+1}^\intercal} \COV(\bar p_{t+1})^{-1}\bar p_{t+1}.
\end{equation}
The one step update of the covariance matrix of the global innovation process  is calculated as
$ \COV(\bar p_{t+1})= \Exp{(\bar y_{t+1} -\Exp{\bar y_{t+1}\mid \mathcal{H}_t^g}) (\bar y_{t+1} -\Exp{\bar y_{t+1}\mid \mathcal{H}_t^g})^\intercal}
=(C_{t+1}+\bar C_{t+1}) \bar \Sigma_{t+1 \mid t} (C_{t+1}+\bar C_{t+1})^\intercal + \COV((S_{t+1}+\bar S_{t+1})\bar v_{t+1})$.
Lastly, from~\eqref{eq:dynamics_deep} and~\eqref{eq:global_etimates},  one arrives at:
$
\Exp{\bar y_{t+1} \mid \mathcal{H}_t^g}= (C_{t+1} +\bar C_{t+1})z_{t+1 \mid t}$.
\end{proof}
From Lemmas~\ref{lemma:Summation} and~\ref{lemma:independent_g}, we obtain 
\begin{multline}\label{eq:update_local_estimate_1}
\Delta{\hat x}^i_{t+1 \mid t+1}=\Exp{\Delta{x}^i_{t+1}\mid \mathcal{H}_{t+1}}
=\Exp{\Delta{x}^i_{t+1}\mid \mathcal{H}_{t}}\\
+ \Exp{\Delta{x}^i_{t+1}\mid  \mathcal{H}^{\perp}_{t+1}}
=\Delta{\hat x}^i_{t+1 \mid t}+ \Exp{\Delta x^i_{t+1} \mid  \mathcal{H}^{\perp}_{t+1}}.
\end{multline}

\begin{Lemma}[Auxiliary update]\label{lemma:local}
The estimate update of the auxiliary  state of agent $i \in \mathbb{N}_n$ given $\mathcal{H}^{\perp}_{t+1}$ is provided by:
\begin{align}
&\Exp{\Delta x^i_{t+1} \mid  \mathcal{H}^{\perp}_{t+1}}= \Exp{\Delta x^i_{t+1} \mid \Delta p^i_{t+1}}\\
&= \Sigma_{t+1 \mid t} C_{t+1}^\intercal \Big( C_{t+1} \Sigma_{t+1 \mid t} C_{t+1}^\intercal + S_{t+1} \Sigma^v_{t+1} S_{t+1}^\intercal\Big)^{-1} \\
&\times (\Delta y^i_{t+1} - C_{t+1} \Delta \hat x^i_{t+1|t}).
\end{align}
\end{Lemma}
\begin{proof}
Let $\Delta \mathbf p_{t+1}=\VEC(\Delta p^1_{t+1},\ldots,\Delta p^n_{t+1})$.  We consider
\begin{multline}
\hspace{-.5cm}\Exp{\Delta x^i_{t+1} |  \mathcal{H}^{\perp}_{t+1}} \substack{(a)\\=}\Exp{\Delta x^i_{t+1} | \Delta \mathbf p_{t+1}} + \Exp{\Delta x^i_{t+1} | \bar{ \mathbf p}_{t+1}}\substack{(b)\\=}
\\
\hspace{-.2cm} \Exp{\Delta x^i_{t+1} | \Delta \mathbf p^{-i}_{t+1}}\substack{(c)\\=}\Exp{\Delta x^i_{t+1} (\Delta \mathbf p^{-i}_{t+1})^\intercal}\COV(\Delta \mathbf p^{-i}_{t+1})^{-1}\hspace{-.1cm}\Delta \mathbf p^{-i}_{t+1},
\end{multline}
where $(a)$ follows from Lemma~\ref{lemma:Summation}; $(b)$ follows from the fact that auxiliary innovation processes are linearly dependent according to~Lemma~\ref{lemma:LD_aux} and auxiliary states are independent of the global innovation process according  to  Lemma~\ref{lemma:Orthogonality_ innovation_processes}, and $(c)$ follows from Lemma~\ref{lemma:Gaussian_orth}. Thus, one can expand the above equation using~Lemma~\ref{lemma:cross_terms} and Proposition~\ref{lemma:inverse} to obtain
\begin{equation}\label{eq:aux_prood_2}
\begin{cases}
\Exp{\Delta x^i_{t+1} |  \mathcal{H}^{\perp}_{t+1}}\hspace{-.1cm}= \hspace{-.1cm}\frac{-\alpha_i}{n}\Sigma_{t+1 |t} C_{t+1}^\intercal \hspace{-.1cm}\otimes\hspace{-.1cm} (\boldsymbol \alpha^{-i})^\intercal (\mathbf I_{(n-1) \times (n-1)} -\\\frac{1}{n}  \boldsymbol \alpha^{-i} (\boldsymbol \alpha^{-i})^\intercal )^{-1}
\hspace{-.1cm} \otimes \hspace{-.1cm} (C_{t+1} \Sigma_{t+1|t}C_{t+1}^\intercal
\hspace{-.1cm}+\hspace{-.1cm} S_{t+1} \Sigma^v_{t+1}S_{t+1}^\intercal \hspace{-.05cm})^{-1}\hspace{-.05cm}\\
\times \Delta \mathbf p^{-i}_{t+1}
=\Sigma_{t+1 |t} C_{t+1}^\intercal \hspace{-.1cm} \otimes \hspace{-.1cm} ( -\frac{\alpha_i}{n}- \frac{1}{n\alpha_i} (\boldsymbol \alpha^{-i})^\intercal \boldsymbol \alpha^{-i})(\boldsymbol \alpha^{-i})^\intercal \\ \otimes (C_{t+1} \Sigma_{t+1|t}C_{t+1}^\intercal
+ S_{t+1} \Sigma^v_{t+1}S_{t+1}^\intercal)^{-1}\Delta \mathbf p^{-i}_{t+1}\\
=\hspace{-.1cm} \Sigma_{t+1 |t} C_{t+1}^\intercal (C_{t+1} \Sigma_{t+1|t}C_{t+1}^\intercal
\hspace{-.1cm}+\hspace{-.1cm} S_{t+1} \Sigma^v_{t+1}S_{t+1}^\intercal)^{-1}\Delta  p^{i}_{t+1},
\end{cases}
\end{equation}
where  $(\boldsymbol \alpha^{-i})^\intercal \Delta \mathbf p^{-i}_{t+1}=-\alpha_i  \Delta  p^{i}_{t+1}$ according to linearly dependent relations in Lemma~\ref{lemma:LD_aux} and~definition of auxiliary innovation processes  in \eqref{eq:innovation_pro_def}. The proof follows now from~\eqref{eq:innovation_pro_def} and~\eqref{eq:aux_prood_2}, on noting that
$
\Exp{\Delta y^i_{t+1}\mid \mathcal{H}^g_t}=C_{t+1} \Delta \hat x^i_{t+1|t}$.
\end{proof}
It can be shown from  Lemma 1 that $\bar \Sigma_{t+1|t+1}$ and $\Sigma_{t+1|t+1}$,  $\forall i \in \mathbb{N}_n$ satisfy the following equations:
\begin{equation}
\begin{cases}
\bar \Sigma_{t+1|t+1}\hspace{-.1cm} = \hspace{-.1cm}\bar \Sigma_{t+1|t+1}- \Exp{ \bar x_{t+1} \mid \mathcal{H}_{t+1}^{\perp}}(\Exp{ \bar x_{t+1} \mid \mathcal{H}_{t+1}^{\perp}})^\intercal=\\
\bar \Sigma_{t+1 \mid t} (C_{t+1}\hspace{-.1cm}+\hspace{-.1cm} \bar C_{t+1})^\intercal ( (C_{t+1}\hspace{-.1cm}+\hspace{-.1cm} \bar C_{t+1}) \bar \Sigma_{t+1 \mid t} (C_{t+1}\hspace{-.1cm}+\hspace{-.1cm}\bar C_{t+1})^\intercal +\\
 \hspace{-.1cm} ( \hspace{-.05cm}S_{t\hspace{-.02cm}+\hspace{-.02cm}1}\hspace{-.1cm}+\hspace{-.1cm}\bar S_{t\hspace{-.02cm}+\hspace{-.02cm}1}\hspace{-.05cm}) \COV(\bar v_{t\hspace{-.02cm}+\hspace{-.02cm}1})(\hspace{-.05cm}S_{t\hspace{-.02cm}+\hspace{-.02cm}1}\hspace{-.1cm}+\hspace{-.1cm}\bar S_{t\hspace{-.02cm}+\hspace{-.02cm}1}\hspace{-.05cm})^\intercal )^{-1}
(\hspace{-.05cm}C_{t+1}\hspace{-.1cm}+\hspace{-.1cm}\bar C_{t+1} \hspace{-.05cm})\bar \Sigma_{t+1 \mid t},\\
\hspace{-.1cm}\Sigma_{t+1 \mid t+1}\hspace{-.1cm}=\hspace{-.1cm} \Sigma_{t\hspace{-.02cm}+\hspace{-.02cm}1|t}\hspace{-.1cm}-\hspace{-.1cm}(1\hspace{-.1cm}-\hspace{-.1cm}\frac{\alpha_i^2}{n}) \Exp{\Delta x^i_{t\hspace{-.02cm}+\hspace{-.02cm}1} \hspace{-.1cm}\mid \hspace{-.1cm} \mathcal{H}_{t\hspace{-.02cm}+\hspace{-.02cm}1}^{\perp}}(\Exp{\Delta x^i_{t+1}\hspace{-.1cm} \mid\hspace{-.1cm} \mathcal{H}_{t+1}^{\perp}})^\intercal\\
\hspace{-.1cm}=\hspace{-.1cm}\Sigma_{t \hspace{-.02cm}+\hspace{-.02cm}1|t} C_{t\hspace{-.02cm}+\hspace{-.02cm}1}^\intercal\hspace{-.05cm} ( C_{t\hspace{-.02cm}+\hspace{-.02cm}1} \Sigma_{t\hspace{-.02cm}+\hspace{-.02cm}1|t} C_{t\hspace{-.02cm}+\hspace{-.02cm}1}^\intercal\hspace{-.1cm}+\hspace{-.1cm} S_{t\hspace{-.02cm}+\hspace{-.02cm}1} \Sigma^v_{t\hspace{-.02cm}+\hspace{-.02cm}1} S_{t+1}^\intercal \hspace{-.05cm})^{-1} C_{t+1}\Sigma_{t+1|t}
\end{cases}
\end{equation}

\subsection{Step 5: Kalman filters}
From equations~\eqref{eq:update_bar_z} and~\eqref{eq:update_local_estimate_1} and update rules  in Lemmas~\ref{lemma:global} and~\ref{lemma:local}, one gets two low-dimensional (scale-free) Kalman filters as follows. For every $i \in \mathbb{N}_n$,
\begin{equation}\label{eq:KF1}
\begin{cases}
\Delta \hat x^i_{t+1|t+1} = \Delta  \hat x^i_{t+1|t}+ L_{t+1}(\Delta y^i_{t+1}- C_{t+1} \Delta \hat x^i_{t+1|t}),\\
   \Delta  \hat x^i_{t+1|t}= A_t  \Delta \hat x^i_{t|t}+ B_t \Delta u^i_t,\\
   \Sigma_{t+1|t+1}= (\mathbf I - L_{t+1} C_{t+1})\Sigma_{t+1|t},\\
    \Sigma_{t+1|t}=A_t \Sigma_{t|t} A_t^\intercal +E_t \Sigma^w_t E_t^\intercal,\\
    L_{t+1}= \Sigma_{t+1|t} C_{t+1}^\intercal (C_{t+1} \Sigma_{t+1|t} C_{t+1}^\intercal +S_{t+1}\Sigma^v_{t+1} S_{t+1}^\intercal)^{-1},
\end{cases}
\end{equation}
with the initial conditions $\Delta \hat x^i_{1|0} =  (1- \frac{\alpha_i}{n}\sum_{j=1}^n \alpha_j)\mu^x$ and
    $\Sigma_{1|0}=\Sigma^x$. Furthermore,
\begin{equation}\label{eq:KF2}
\begin{cases}
z_{t+1|t+1} = z_{t+1|t}+ \bar L_{t+1}(\bar y_{t+1}- (C_{t+1}+\bar C_{t+1}) z_{t+1|t}),\\
   z_{t+1|t}= (A_t+\bar A_t)  z_{t|t}+ (B_t+\bar B_t) \bar u_t,\\
  \bar  \Sigma_{t+1|t+1}= (\mathbf I - \bar L_{t+1} (C_{t+1}+\bar C_{t+1}))\bar \Sigma_{t+1|t},\\
 \bar    \Sigma_{t+1|t}=(A_t+\bar A_t) \bar \Sigma_{t|t} (A_t+\bar A_t)^\intercal \\
\qquad \qquad  +\frac{1}{n}(E_t+\bar E_t)\Sigma^w_t(E_t+\bar E_t)^\intercal,\\
  \bar   L_{t+1}= \bar \Sigma_{t+1|t} (C_{t+1}+\bar C_{t+1})^\intercal ((C_{t+1}+\bar C_{t+1}) \bar \Sigma_{t+1|t} \\
  \times (C_{t+1}\hspace{-.1cm}+ \hspace{-.1cm}\bar C_{t+1})^\intercal \hspace{-.1cm}+\hspace{-.1cm}\frac{1}{n}(S_{t+1}+\bar S_{t+1}) \Sigma^v_{t+1}(S_{t+1}\hspace{-.1cm}+\hspace{-.1cm}\bar S_{t+1})^\intercal)^{-1},
\end{cases}
\end{equation} 
with the conditions $z_{1|0} =(\frac{1}{n}\sum_{j=1}^n \alpha_j) \mu^x$ and 
    $\bar \Sigma_{1|0}=\frac{1}{n}\Sigma^x$.
    \section{Main results}\label{sec:main}
We impose the following standard assumption.
\begin{Assumption}
Let matrices $Q_t$ and $Q_t+\bar Q_t$ be symmetric and positive semi-definite and matrices $R_t$ and $R_t+\bar R_t$ be symmetric and positive definite for every time $t \in \mathbb{N}$.
\end{Assumption}
For any  $ t \in \mathbb{N}_{T-1}$, define the following Riccati equations:
\begin{equation}\label{eq:Riccati}
\begin{cases}
P_{t}= Q_t+ A_t^\intercal P_{t+1} A_t - A_t^\intercal P_{t+1} B_t (B_t^\intercal P_{t+1} B_t+ R_t)^{-1}\\
 \times  B_t^\intercal P_{t+1} A_t, \\
\bar P_{t}= Q_t+\bar Q_t+ (A_t+\bar A_t)^\intercal \bar P_{t+1} (A_t+\bar A_t) - (A_t+\bar A_t)^\intercal \\
\times \bar P_{t+1} (B_t+\bar B_t) ((B_t+\bar B_t)^\intercal \bar P_{t+1} (B_t+\bar B_t)+ R_t \hspace{-.05cm}+\hspace{-.05cm} \bar R_t)^{-1}\\
 \times  (B_t+\bar B_t)^\intercal \bar P_{t+1} (A_t+\bar A_t), 
\end{cases}
\end{equation}
where $P_T=Q_T$ and  $\bar  P_T=Q_t+\bar Q_T$.
\vspace{-.1cm}
\subsection{Solution of Problem~\ref{prob1}}
\begin{Theorem}[Decentralized estimation, optimal control,  and separation principle]\label{thm}
Let Assumption~1 hold. The optimal strategy of agent $i \in \mathbb{N}_n$ at time $t \in \mathbb{N}$ under the decentralized  information structure IDSS is given by
\begin{equation}\label{eq:optimal_strategy}
u^{\ast,i}_t= \theta^\ast_t \hat x^i_{t|t} + \alpha_i(\bar \theta^\ast_t - \theta^\ast_t)  z_{t|t},
\end{equation}
 where local and global estimates are computed by two  scale-free Kamlan filters in~\eqref{eq:KF1} and~\eqref{eq:KF2} as follows:
\begin{equation} 
\begin{cases}
\hat x^i_{t+1|t}\hspace{-.1cm}=\hspace{-.1cm} \Exp{x^i_{t+1} \hspace{-.1cm}\mid \hspace{-.1cm} \mathcal{H}_{t}}= A_t \hat x^i_{t|t}+ B_t u^i_{t}+ \alpha_i( \bar A_t z_{t|t}+\bar B_t \bar u_t),\\ 
\hat x^i_{t+1|t+1}\hspace{-.1cm}=\hspace{-.1cm} \Exp{x^i_{t+1} \hspace{-.1cm} \mid  \hspace{-.1cm}\mathcal{H}_{t+1}}\hspace{-.1cm}=\hspace{-.1cm} \hat x^i_{t+1|t} \hspace{-.1cm} + \hspace{-.1cm} L_{t+1}(y^i_{t+1} - C_{t+1} \hat x^i_{t+1|t} )\\
\qquad +\alpha_i(\bar L_{t+1} - L_{t+1})(\bar y_{t+1}- (C_{t+1}+\bar C_{t+1}) z_{t+1|t}),\\
z_{t+1|t}= \Exp{\bar x_{t+1}\hspace{-.1cm}\mid \hspace{-.1cm}\mathcal{H}_{t}}= (A_t+\bar A_t) z_{t|t}+ (B_t+ \bar B_t)\bar u_t,\\ 
 z_{t+1|t+1}= \Exp{\bar x_{t+1} \hspace{-.1cm}\mid \hspace{-.1cm} \mathcal{H}_{t+1}}=z_{t+1|t} \\
\qquad + \bar L_{t+1} (\bar y_{t+1}- (C_{t+1}+\bar C_{t+1}) z_{t+1|t}).
\end{cases}
\end{equation}
In addition, the local and global gains $\theta^\ast_t$ and $\bar \theta_t^\ast$ are obtained by two scale-free Riccati equations in~\eqref{eq:Riccati} as
\begin{equation}
\begin{cases}
 \theta^\ast_t:=-(B_t^\intercal P_{t+1} B_t + R_t)^{-1} B_t^\intercal  P_{t+1} A_t,\\
\bar \theta^\ast_t:=-( (B_t+\bar B_t)^\intercal \bar P_{t+1} (B_t+\bar B_t) + R_t+ \bar R_t)^{-1}\\
\qquad \times (B_t+\bar B_t)^\intercal \bar P_{t+1} (A_t+\bar A_t).
\end{cases}
\end{equation}
\end{Theorem}
\begin{proof}
The proof follows by using the completion-of-square method on the transformed cost function in Lemma~\ref{lemma:cost} and replacing the state values with their conditional expectations, defined in~\eqref{eq:global_etimates} and~\eqref{eq:local_etimates}. On the other hand, we showed that the   conditional expectations could be  computed recursively by Kalman filters in~\eqref{eq:KF1} and~\eqref{eq:KF2},  irrespective of the control strategy (also known as the separation principle~\cite{aoki1967optimization,aastrom2012introduction,Witsenhausen1971separation,bar1974dual}). Note that the separation principle is weaker than the certainty equivalence principle~\cite{van1981certainty}.  The remaining problem is an optimal LQ deep structured team with (perfect) deep state sharing whose solution  is obtained from the  Riccati equations in~\eqref{eq:Riccati}; see~\cite{Jalal2020CCTA,  Jalal2021CDC_MPC,Jalal2019risk}  for more details.
\end{proof}

From Theorem~\ref{thm}, the optimal solution of agent $i\in \mathbb{N}_n$ can be  implemented in a distributed manner. In particular, prior to the operation of the system, agent~$i$ solves two scale-free Riccati equations and Kalman filters to obtain $\{L_{1:T}, \bar L_{1:T}, \theta^\ast_{1:T}, \bar \theta^\ast_{1:T}\}$. At any time  $t \in \mathbb{N}_T$, agent $i$ estimates its local state $x^i_{t}$ by $\hat x^i_{t|t}$ and global state $\bar x_{t}$ by $z_{t|t}$ based on local and global noisy observations $y^i_{t}$ and $\bar y_{t}$. Then, given its private (influence factor) $\alpha_i$, agent~$i$ calculates its optimal strategy according to~\eqref{eq:optimal_strategy}.

\begin{Remark} 
\emph{The only information  shared among agents at each time instant $t$ is noisy deep state $\bar y_{t} \in \mathbb{R}^{d_y}$ whose size (dimension)  is independent of the number of agents $n$. }
\end{Remark}
Establishing Theorem~\ref{thm} for the special case of infinite population (i.e. $n=\infty$) is straightforward because  auxiliary  and global (deep)  variables are mutually orthogonal according to Lemma~\ref{lemma:Relations_primitive_rv} and  Remark~1. This makes it significantly easier to develop a low-dimensional Kalman filter for~$n~=~\infty$.   The same argument holds for the finite-population  case  where deep state is observed perfectly (i.e. $\bar x_t \in \mathcal{H}_t$); see mean-field teams in~\cite{arabneydi2016new,JalalCDC2015}  that considers such a  special case for homogeneous influence factors, where $\alpha_i=1$, $\forall i \in \mathbb{N}_n$.
\subsection{Asymptotic approximation for Problem~\ref{prob1}}

\begin{Theorem}\label{thmm}
As $n \rightarrow \infty$, covariance  matrices of  the global Kalman filter~\eqref{eq:KF2} converge  to zero at rate $1/n$. As a result, the followings hold at  any $t \in \mathbb{N}_T$  for the model with  $n=\infty$.
\begin{itemize}
\item \textbf{Blind optimal global estimator:} The global state  $\bar x_{t}$ is almost surely equal to its conditional  and  unconditional expectation, i.e.  $\bar x_{t} \hspace{-.05cm}= \hspace{-.05cm} z_{t|t} \hspace{-.05cm}= \hspace{-.05cm}\Exp{\bar x_t}$, which may be viewed  as the~certainty~equivalence approximation~\cite{van1981certainty} and (weighted) mean-field approximation~\cite{parisi1988statistical},~respectively.
\item \textbf{Optimal fully decentralized strategy:}  Agent~$i \in \mathbb{N}_n$  requires access only to  its local observation~$y^i_t$  to calculate~\eqref{eq:optimal_strategy}, leading to a fully decentralized strategy, where  the global observation is perfectly  predicted as follows: $\bar y_{t+1}=\lim_{n \rightarrow \infty} \frac{1}{n}\sum_{i=1}^n\alpha_i y^i_{t+1}=(C_{t+1}+\bar C_{t+1}) \lim_{n \rightarrow \infty}\bar x_{t+1}=(C_{t+1}+\bar C_{t+1}) \bar z_{t+1|t}$.
\end{itemize} 
\end{Theorem}
\begin{proof}
The proof follows by noting that the covariance matrix of the summation of  any $n$ uniformly bounded independent random variables goes to zero with the rate $1/n$, and that $z_{t+1|t+1}=z_{t+1|t}$ for $n=\infty$, according to~\eqref{eq:KF2}.
\end{proof}

\subsection{Extension to least-square estimation}
The results of Theorems~\ref{thm} and~\ref{thmm} naturally hold for the best linear strategy minimizing  the least-square estimation error without imposing  any  Gaussian assumption~\cite{simon2006optimal}. In such a case,  the proof uses
 Hilbert space analysis with  inner product $\langle x,y\rangle= \Exp{xy^\intercal}$, where $x$ and $y$ are random variables.

\section{Conclusions}\label{sec:conclusions}
 In this paper,  a new class of large-scale decentralized multi-agent systems, called deep structured teams, with noisy measurements was studied. A novel transformation-based approach was proposed  to  introduce a bi-level orthogonal relationship between the agents across both state space and  time horizon. The optimal  solution was shown to be  linear in the local and global estimates and  computed by two standard decoupled scale-free backward and forward equations (i.e., Riccati equations and Kalman filters). In addition, a fully decentralized sub-optimal strategy was developed,   whose performance converges to that of the optimal one at a rate inversely proportional to the number of agents.  The main results of this paper naturally extend to a model with multiple sub-populations and multiple features in a fashion similar to the one proposed in~\cite[Section IV]{Jalal2021CDC_MPC} and~\cite{Jalal2019risk}.
\bibliographystyle{IEEEtran}
\bibliography{Jalal_Ref}
\end{document}